\documentclass[a4paper]{article}

\usepackage{amssymb, amsmath, amsthm, mathtools, enumitem}
\usepackage{a4wide}

\usepackage[colorlinks=true]{hyperref}
\usepackage{todonotes}

\newcommand\g{\ensuremath{\mathbf{g}}}
\newcommand\h{\ensuremath{\mathbf{h}}}
\newcommand\I{\ensuremath{(0, 1]}}
\newcommand\x{\ensuremath{\mathbf{x}}}
\DeclareMathOperator\grad{grad}
\DeclareMathOperator\supp{supp}

\newcommand\bel[1]{\begin{equation}\label{#1}}
\newcommand\ee{\end{equation}}
\newcommand\bital{\begin{enumerate}[label = {(\roman*)}, ref = {\roman*}, itemsep = 2mm] \vskip .2cm}
\newcommand\eital{\end{enumerate}}

\numberwithin{equation}{section}

\theoremstyle{plain}

\newtheorem{definition}{Definition}[section]

\newtheorem{Theorem}[definition]{Theorem}

\newtheorem{Lemma}[definition]{Lemma}
\newtheorem{Corollary}[definition]{Corollary}
\newtheorem{Conjecture}[definition]{Conjecture}

\theoremstyle{remark}
\newtheorem{remark}[definition]{Remark}

\newcommand{\R}{\ensuremath{\mathbb{R}}}
\newcommand{\N}{\ensuremath{\mathbb{N}}}
\newcommand{\G}{{\mathcal{G}}}
\newcommand{\eps}{\ensuremath{\varepsilon}}
\newcommand{\D}{\mathcal{D}}
\newcommand{\Cinf}{\ensuremath{\mathcal{C}^\infty}}
\newcommand{\E}{\ensuremath{{\mathcal E}}}
\newcommand{\EM}{\ensuremath{{\mathcal E}_{\mathrm{M}}}}
\newcommand{\NN}{\ensuremath{{\mathcal N}}}

\begin{document}
\title{The failure of the Ehlers--Kundt conjecture\\ in the impulsive case} 
\author{
 Moriz L. Frauenberger$^1$\thanks{{\tt a1203035@univie.ac.at}},\\
 James D. E. Grant$^2$\thanks{{\tt j.grant@surrey.ac.uk}},\\ 
 Roland Steinbauer$^1$\thanks{{\tt roland.steinbauer@univie.ac.at}}\\ \\
 $^1$Faculty of Mathematics, University of Vienna, \\
 Oskar-Morgenstern-Platz 1, 1090 Vienna, Austria. \\ \\
 $^2$School of Mathematics and Physics, University of Surrey,\\
 Guildford, GU2 7XH, UK. 
 }

\maketitle

\begin{abstract}
In 1962, Ehlers and Kundt conjectured that plane waves are the only class of complete Ricci-flat~\emph{pp}-waves, i.e.\ metrics on ${\mathbb R}^4$ of the form
\[ 
ds^2=2du\,dv+dx^2+dy^2+H(x,y,u)du^2\,.
\]
Recently, Flores and S\'{a}nchez gave a proof of the conjecture in the fundamental case of spatially polynomially bounded profile functions $H$. However, \emph{impulsive} \emph{pp}-waves, i.e., waves with concentrated profile functions of the form $H(x,y,u)=f(x,y)\,\delta(u)$ ($\delta$, the Dirac measure) have been found to be complete. We summarise completeness results for several classes of impulsive wave spacetimes achieved during the last years and discuss them in the context of the Ehlers--Kundt conjecture.
\bigskip

\medskip

 
\end{abstract}

\newpage
\section{Introduction}
The class of~\emph{pp-waves\/} is an important family of exact solutions to the Einstein equations. 
\begin{definition}[pp-wave]
\label{def:ppWave}
A~\emph{plane-fronted wave with parallel rays\/} is a Lorentzian manifold $(M, \g)$ with $M = \R^4 = \{ (u, v, \x) \mid u, v \in \R, \x = (x^1, x^2) \coloneqq (x, y) \in \R^2 \}$ with metric given by 
\begin{equation}
\label{eq:pp_metric}
\g = 2 \, du \, dv + H(u, \x) \, du^2 + \delta_{ij} \, dx^i \, dx^j,
\end{equation}
where $H \colon \R \times \R^2 \to \R$ is a smooth function.
\end{definition}
Such metrics first arose in the class studied by Brinkmann~\cite{B1925}, who investigated the problem of when an Einstein metric can be conformally mapped to another Einstein metric, and, slightly later, in the physics literature in the work of Baldwin--Jeffrey~\cite{BJ}. Unfortunately, these works do not seem to have attracted great attention at the time and, indeed, it is only after the discovery of such metrics by Peres~\cite{Peres1959}, 
H\'{e}ly~\cite{H59} and Kundt~\cite{Kundt1961} (see, also, \cite{BPR}) that the work of Brinkmann appears to have been recognised.%
\footnote{The first mention of Brinkmann's work that we are aware of is~\cite{Kundt1961}, while the work of Baldwin--Jeffrey appears to have been rediscovered in the 1970's. A discussion of the history of these developments can be found in~\cite{Robinson2019}.} %
The metrics~\eqref{eq:pp_metric} were first referred to as pp-waves in the 1962 paper of Ehlers and Kundt~\cite{EK62}.%
\footnote{Some historical information on the topic can be found in~\cite[Chapter~2.5]{EK62}. For a modern discussion, see~\cite[Section~1.1]{RAC2023}.} 

The Ricci tensor of the metric~\eqref{eq:pp_metric} is 
\[
\mathbf{Ric}_{\g} = -\frac{1}{2} \left( \frac{\partial^2 H}{\partial x^2} + \frac{\partial^2 H}{\partial y^2} \right) du \otimes du. 
\]
Given the form of the metric $\g$, it follows that the scalar curvature of $\g$ is identically zero. Therefore, the metric~\eqref{eq:pp_metric} is Ricci flat and satisfies the vacuum Einstein equations if and only if the function $H$ is harmonic in the spatial directions. 

A special class of pp-metrics are the so-called~\emph{plane waves\/} for which, for each fixed $u \in \R$, the map $\x \mapsto H(u, \x)$ is a quadratic form, i.e. $H(u, \x) = h_{ij}(u) x^i x^j$ where, without loss of generality, we take the matrix $h = \left( h_{ij} \right)$ to be symmetric. In this case, the metric~\eqref{eq:pp_metric} is Ricci-flat if and only if the matrix $h$ is trace-free, i.e. $h_{ii} = 0$. 

More generally, it can be shown~\cite[Lemma~3.1]{GL} that given any $(n+2)$-dimensional Lorentzian metric that admits a parallel null vector field $Z$ with the property that the curvature tensor, viewed as a linear map from bivectors to bivectors, is non-trivial but zero on restriction to $Z^{\perp} \wedge Z^{\perp}$,%
\footnote{Here $Z^{\perp} \coloneqq \{ X \mid \g(Z, X) = 0 \}$ is the annihilator of the vector field $Z$.} %
there exist local coordinates $(u, v, x^1, \dots, x^n)$ in which the metric is locally of the form~\eqref{eq:pp_metric}. Imposing that $\nabla_U R = 0$ for all $U \in Z^{\perp}$ then yields (locally) the plane wave metrics, where the function $H(u, \x)$ is of the form $h_{ij}(u) x^i x^j$. 

\begin{remark}
The causal behaviour of plane wave metrics is quite well understood. It was shown by Penrose~\cite{Penrose_Remarkable} that plane wave metrics are~\emph{not\/} global hyperbolic. This result was refined in~\cite{EE1},%
\footnote{See, in particular, Theorem~6.9 and Proposition~5.12 in~\cite{EE1}.} %
where it was shown that (Ricci-flat) plane wave metrics are strongly causal, stably causal and causally continuous, but they are~\emph{not\/} causally simple (and, hence, not globally hyperbolic). 
\end{remark}

\subsection{$N$-fronted waves} 
We will pose our results on a generalisation of the above class of metrics introduced in~\cite{CFS2003} (see also~\cite{FS2003}), where we allow for wave surfaces that are described by a Riemannian manifold $(N, \h)$. We follow the terminology of~\cite{Saemann2012}, referring to such as $N$-fronted waves with parallel rays.

\begin{definition}[NPW]
\label{def:NPW}
An~\emph{$N$-fronted wave with parallel rays (NPW)\/} is a spacetime $(M, \g)$ where $M = N \times \R^2_1$ for an $(n-2)$-dimensional Riemannian manifold $(N, \h)$ such that the metric 
is given by
\begin{equation}
\label{eq:NPW_metric}
\g = 2 du \, dv + H(u, x) \, du^2 + \h,
\end{equation}
where $H \colon \R \times N \to \R$ is a smooth function.%
\footnote{Here, $u$ and $v$ will be global coordinates on $\R^2$. We will often denote local coordinates on $N$ by $x \equiv (x^1, \dots, x^{n-2})$.}
\end{definition} 

The Ricci tensor of the NPW metric~\eqref{def:NPW} is
\begin{equation}
\label{eq:R_ij}
\mathbf{Ric}_{\g} = \left( - \frac{1}{2} \bigtriangleup_{\h} H \right) du \otimes du + \mathbf{Ric}_{\h}, 
\end{equation}
where $\mathbf{Ric}_{\h}$ is the Ricci tensor of the $(n-2)$-dimensional metric $\h$, and $\bigtriangleup_{\h}$ is the Laplace--Beltrami operator of the metric $\h$. Therefore, an NPW satisfies the vacuum Einstein equations if and only if $(N, \h)$ is Ricci-flat and the function $H$ is spatially harmonic, i.e. 
\[
\bigtriangleup_{\h} H = 0.
\]

\begin{remark}
The results of Ehrlich and Emch~\cite{EE1} were extended and significantly strengthened for NPW metrics in~\cite{FS2020}, where it was shown that, if the manifold $N$ is non-compact, then quadratic behaviour of the function $H(u, x)$ asymptotically is critical (cf.\ the plane wave metrics). More specifically, it was shown that: 
\bital
\item If $-H(u, x)$ behaves sub-quadratically as $x \to \infty$, and the metric $\h$ is complete, then $(M, \g)$ is globally-hyperbolic; 
\item If $-H(u, x)$ grows at most quadratically as $x \to \infty$, and the metric $\h$ is complete, then $(M, \g)$ is strongly causal; 
\item Otherwise, $(M, \g)$ is causal, but may be non-distinguishing. 
\eital
\end{remark}

\subsection{Geodesics}
The geodesic equations for a curve $c \colon I \to M; s \mapsto (u(s), v(s), \x(s))$ in an NPW $(M, \g)$ are
\begin{align*}
\frac{d^2 u}{ds^2} &= 0, 
\\
\frac{d}{ds} \left( \frac{dv}{ds} + H \, \frac{du}{ds} \right) &= \frac{1}{2} \, \frac{\partial H}{\partial u} \, \left( \frac{du}{ds} \right)^2, 
\\
\nabla^{\h}_{\x'} \x' &= \frac{1}{2} \,\left( \frac{du}{ds} \right)^2 \, \grad^{\h} H. 
\end{align*}	
Solving the first of these equations, 
\begin{equation}
\label{usolution}
u(s) = a s + b, \qquad s \in \R, 
\end{equation}
where $a, b$ are constants. We then require that 
\begin{subequations}
\begin{gather}
\frac{d^2v}{ds^2} + \frac{a^2}{2} \frac{\partial H}{\partial u} + a \, \frac{\partial H}{\partial x^i} \, \frac{dx^i}{ds} = 0, 
\label{veqn}
\\
\nabla^{\h}_{\x'} \x' = \frac{a^2}{2} \, \grad^{\h} \! H. 
\label{xeqn}
\end{gather}
\label{geodeqns}\end{subequations}
If $a=0$, then $u(x) = u(0)$, so the geodesics are parallel to the wavefront, while $v(s) = c s + d$ and the curve $s \mapsto \x(s) \in N$ is a geodesic in $(N, \h)$. 

If $a \neq 0$, then $\frac{du}{ds} = a \neq 0$, so we may parametrise the geodesics by $u$. Equation~\eqref{xeqn} then becomes the equation for motion on $(N, \h)$ under a potential $V \coloneqq - \frac{1}{2} H$. Once one has solved this equation for $\x(u)$, one would then insert this solution into~\eqref{veqn}, and solve for $v(u)$.%
\footnote{In practice, the condition that a geodesic satisfy $\g(\gamma', \gamma') = \lambda$, gives an equation for $v'$ which, at least for plane waves, one can then integrate~\cite{EE1}.} 

\subsection{The Ehlers--Kundt conjecture}
In their 1962 article~\cite[pp.~97]{EK62}, Ehlers and Kundt proposed the following problem, suggesting that plane waves play a privileged role among pp-waves: 

\bigskip

{\small ``Prove the plane waves to be the only $g$-complete pp waves, no matter which topology one chooses.''}

\bigskip

\noindent{}This statement is now referred to as ``the Ehlers--Kundt conjecture.'' Intuitively speaking, a complete, Ricci-flat pp-wave describes a gravitational field without singularities (i.e. it is geodesically complete) and without internal matter sources (i.e. it satisfies the vacuum Einstein equations). By analogy with Maxwell’s equations in vacuum, where the simplest solutions are monochromatic plane waves, the Ehlers--Kundt conjecture asserts that the gravitational plane waves play the corresponding role as the most elementary, globally regular, source-free gravitational fields.

\medskip
In a more precise formulation, noting that only Ricci-flat metrics are considered in~\cite{EK62}, the conjecture can be stated as follows.
\begin{Conjecture}[Ehlers--Kundt]
A Ricci-flat pp-wave $(\R^4, g)$, with metric 
\begin{equation*}
\g = 2 du \, dv + H(u, \x) \, du^2 + dx^2 + dy^2,
\end{equation*}
where $H \colon \R \times \R^2 \to \R$, is geodesically complete if and only if the function $\x \mapsto H(u, \x)$ is polynomial of degree at most two, for all $u \in \R$.
\end{Conjecture} 
\begin{remark}
\label{rem:pwgc}
The geodesic equations for the plane wave metrics can be integrated directly,%
\footnote{Up to integration of a linear second-order system of ordinary differential equations.} %
and show that plane wave the metrics are geodesically complete. (See~\cite[Theorem~2-5.9]{EK62} or~\cite[Chapter~13]{BEE}.)
\end{remark}

While quadratic behavior of $H$ has been successfully shown to be critical for various geometric properties of the pp-wave spacetime, e.g.\ asymptotic flatness and causality (cf.~\cite{FS2003}), the Ehlers--Kundt conjecture remains an open problem (cf.~\cite[Section 5]{RAC2023}). 

\subsection{Positive results}
A special case of the Ehlers--Kundt conjecture, the so-called~\emph{polynomial Ehlers--Kundt conjecture}, in which $H$ is assumed to be bounded by a polynomial, has been solved by Flores and S\'{a}nchez in~\cite{FS2020}. They work in the setting of dynamical systems and, more precisely, assume $V = - \frac{1}{2}H$ to be polynomially $u$-bounded.

\begin{definition}
A function $V \colon \R^2 \times \R \to \R$ is~\emph{polynomially $u$-bounded\/} if, for all $u \in \R$, there exists an $\eps_0$ and a polynomial $q_0 \in \R[x,y]$ such that, for all $(x, y, u) \in \R^2 \times (u_0-\eps_0, u_0 + \eps_0)$, we have 
\[
V(x,y,u) \leq q_0(x,y).
\]
\end{definition}

With this notion at hand, we can state the polynomial Ehlers--Kundt conjecture as follows.

\begin{Theorem}[Polynomial Ehlers--Kundt Conjecture]
Let $(\R^4, \g)$ be a pp-wave with the properties that 
\bital
\item the function $V = - \frac{1}{2}H$ is in $\mathcal{C}^1(\R^2 \times \R)$ and polynomially $u$-bounded; 
\item for all $u \in \R$, $\x \mapsto V(\x, u)$ is in $\mathcal{C}^2(\R^2)$ and harmonic. 
\eital
\medskip
Then $(\R^4, \g)$ is geodesically complete if and only if $V(\cdot,\cdot,u)$ is an at most quadratic polynomial for all $u \in \R$, i.e.\ $(\R^4,g)$ is a plane wave.
\end{Theorem}


If $V$ is autonomous and homogeneous (i.e.\ all terms are of the highest degree) then, by a suitable choice of polar coordinates, it can be written as
\[
V(\rho, \theta) = - \rho^n \cos(n \theta).
\]
In this case, it can be shown that there exist neighbourhoods $D_k[\rho_0,\frac{\pi}{2n}]$ of the radial trajectories $\gamma_k(s) = (\rho(s), \theta_k)$ for $\theta_k = \frac{2\pi k}{n}$, such that trajectories starting in $D_k[\rho_0,\frac{\pi}{2n}]$ stay within this region and are incomplete.

For the non-homogeneous and non-autonomous case, the proof works analogously by constructing appropriate radial curves $\gamma_k$ and regions $D[\rho_0,\theta_+]$ with the corresponding properties.

\medskip

As stated in~\cite[pp.~7509]{FS2020}, and briefly alluded to above, the natural setting for the Ehlers--Kundt conjecture (and for the more general problems discussed in~\cite{FS2020}) is where the potential function $V \equiv - \frac{1}{2} H$ is $C^1$ in $u$, in order that the geodesic equations for the metric $\g$ make sense classically, and $C^2$ in the spatial variables, in order that the spatial-harmonicity condition on $H$ be classically well-defined. What we will show is that, for~\emph{impulsive\/} waves, where the function $H$ is not continuous in $u$, the Ehlers--Kundt conjecture fails.

\section{Completeness for impulsive waves}

\subsection{Impulsive gravitational waves}

In this section, we turn to~\emph{impulsive\/} waves or, more precisely, to the impulsive version of the metric~\eqref{eq:NPW_metric}. Generally, impulsive gravitational waves model short but violent pulses of gravitational or other radiation. They were introduced by Penrose in the late 1960's and considered in detail in~\cite{Pen:72}. Over the decades the models have been widely generalised to include a non-vanishing cosmological constant as well as gyratonic terms~\cite[Ch.~20]{GP:09}. They have been found to arise as ultrarelativistic limits of Kerr--Newman and other static spacetimes (see, e.g.,~\cite{AS:71,HT:93,LS:94,Bal:97}), which makes them interesting models for quantum scattering in general relativistic spacetimes (see, e.g.,~\cite{Bla:11,S:18}), the wave memory effect~\cite{ZDH:impulsive,S:19} and entanglement harvesting (see~\cite{GKMTT:21} and references therein). Moreover, they have found applications in astrophysics~\cite{BH:03}, and serve as key examples in mathematical investigations of low-regularity spacetimes. (See~\cite{SV:06} and forthcoming works in synthetic Lorentzian geometry.)

Following~\cite{P:02}, the simplest way to introduce impulsive waves in our context is to start with metric~\eqref{eq:NPW_metric} and set
\medskip
\bel{ieprofile}
H(x, u) = f(x) \delta(u) \,,
\ee

\medskip
\noindent{}where $\delta$ denotes the Dirac function and $f$ is a smooth function on $N$, see~\cite{SS:12,SS:15}. (Note that the Ricci tensor of the metric does not include $u$-derivatives of $H$ (see~\eqref{eq:R_ij}), so the field equations do not restrict its $u$-behavior.) Therefore, we consider the spacetime $M = N \times \R^2_1$ and metric~\eqref{eq:NPW_metric} with $H$ as in~\eqref{ieprofile}. More explicitly we define the following.

\begin{definition}[INPW]
\label{def:INPW}
An~\emph{impulsive NPW\/} (or~\emph{INPW}) is a spacetime $(M, \g)$ with $M = N \times \R^2_1$, $(N, \h)$ an $(n-2)$-dimensional Riemannian manifold and
with metric
\begin{equation}
\label{eq:INPW_metric}
\g = 2 du \, dv + f(x)\delta(u) \, du^2 + \h,
\end{equation}
where $f \colon N \to \R$ is a smooth function and $\delta$ is the Dirac function.
\end{definition}

While we can view this metric as a(n impulsive) limit of sandwich waves with ever shorter but stronger profile function, it clearly is outside of the Geroch--Traschen class~\cite{GT:87}, the largest class of metrics that allow for a stable definition of distributional curvature~\cite{LFM:07,SV:09}. The high degree of symmetry allows one to calculate the curvature of the metric~\eqref{eq:INPW_metric} directly in distributions without encountering any ill-defined products to obtain, 
\begin{equation}
\mathbf{Ric}_{\g} = - \frac{1}{2}\, \delta(u)\, \bigtriangleup_{\h} f(x) + \mathbf{Ric}_{\h}, 
\end{equation}
(cf.~\eqref{eq:R_ij}). Nevertheless, for the sake of mathematical rigour, we will study~\emph{regularisations\/} of the metric. In addition, however, the form of geodesic equations~\eqref{geodeqns} with $H$ of the form~\eqref{ieprofile} (or~\eqref{geodeqns-delta-eps} below) implies that there is no distributional solution concept available for the geodesic equations of the metric~\eqref{eq:INPW_metric}. Moreover, this approach via regularisation is consistent with our view of~\eqref{eq:INPW_metric} as an impulsive limit.

Note that, at least for the special case of impulsive pp-waves, one can explicitly find a so-called ``continuous form'' of the metric. In particular, Penrose argued in~\cite{Pen:72} that the distributional Brinkmann form of impulsive waves and its Lipschitz continuous Rosen form are related by a ``discontinuous coordinate transformation''. Clearly, the latter metric~\emph{does\/} lie in the Geroch--Traschen class, thereby allowing a classical treatment of its curvature and the field equations. A mathematically rigorous account of this transformation is given in~\cite{KS:99a}, and recently for the case of $\Lambda \neq 0$ in~\cite{SSSS:24}. In addition, the ``continuous form'' of the metric is, in fact, locally Lipschitz continuous. This is essential when one studies its geodesics~\cite{Ste:14,LLS:21} and its causality~\cite{CG:12}. Indeed, Lipschitz regularity has turned out to be the threshold for classical causality theory to hold, and several key properties fail below it~\cite{GKSS:20}. (However, see~\cite{Lin:25}). Also, there are hints that an additional bound on the curvature in $L^p$ ($p>2\dim(M)$) allows one to extend standard causality to metrics in $W^{2,p}$ using the RT-transformations of~\cite{RT:22,RT:25}.

\subsection{Regularisation and the geodesic equation}

We now introduce the regularisation approach mentioned above, which we will use throughout the remainder of the paper. We first introduce a very general regularisation of the Dirac delta. 

\begin{definition}
\label{def:StrictDeltaNet}
A net $(\delta_{\eps})_{\eps \in \I}$ of compactly supported smooth functions, is called a~\emph{strict delta net\/} if: 
\begin{enumerate}
\item $\supp(\delta_{\eps}) \to \{0\}$ as $\eps \to 0$; 
\item $\int_{\R} \delta_{\eps}(x) \, dx \to 1$ as $\eps \to 0$; 
\item $\exists K > 0$ such that $\forall \eps \in \I$, $\int_{\R} |\delta_{\eps}(x)| \, dx \le K$.
\end{enumerate}
\end{definition}

We clearly have that $\delta_\eps \to \delta$ weakly as $\eps\to 0$. More generally, one may use convolution with $\delta_\eps$ to regularise distributions $u \in \mathcal{D}'(\R)$ via
\begin{equation}\label{eq:conv}
u*\rho_\eps(x) \coloneqq \langle u(x-y),\delta_\eps(y)\rangle, 
\end{equation}
where $\langle\ ,\ \rangle$ denotes the distributional action (in the variable $y$). (We will follow the notation of~\cite{FJ:98} for distribution theory.) For each $\eps \in I$, $u*\rho_\eps$ is a smooth function (of $x$). This family of functions weakly converge to $u$ as $\eps \to 0$. 

\medskip
We may now define the regularised spacetimes that we will use. 

\begin{definition}[rINPW]
\label{def:rINPW}
A~\emph{regularised impulsive NPW\/} (or~\emph{rINPW}) is a spacetime $(M,\g_\eps)$ with $M = N \times \R^2_1$, $(N, \h)$ an $(n-2)$-dimensional Riemannian manifold, with metric
\begin{equation}
\label{eq:rINPW_metric}
\g_\eps = 2 du \, dv + f(x)\delta_\eps(u) \, du^2 + \h,
\end{equation}
where $f \colon N \to \R$ is a smooth function, $\delta_\eps$ is a model delta net, and $\eps\in \I$.%
\footnote{For brevity, we will use the term rINPW for both a specific spacetime $(M, \g_\eps)$ for some fixed $\eps \in \I$, as well as for the entire family of spacetimes $(M,\g_\eps)$ with $\eps \in \I$.}
\end{definition}

The geodesic equations for~\eqref{eq:rINPW_metric} are simply the regularisations of the geodesic equations of the metric~\eqref{eq:INPW_metric}. As above, we can entirely dispense with the $u$-equation: If $a=0$ the geodesic is parallel to the impulsive wave surface at $u=0$, and therefore never meets the wave. Moreover, for sufficiently small $\eps$, the geodesic will be outside of the regularisation region given by the support of $\delta_\eps$. Therefore, these geodesics reduce to geodesics of the `background` flat spacetime without any wave. In the case $a \neq 0$, on the other hand, we may use $u$ as an affine parameter along the geodesic and solely work with the $\x$- and $v$-equations
\begin{subequations}
\begin{gather}
\frac{d^2v}{ds^2} + \frac{a^2}{2} f \, \dot\delta_\eps(u) + a \, \frac{\partial f}{\partial x^i} \, \frac{dx^i}{ds} \delta_\eps(u) = 0, 
\label{veqn-eps}
\\
\nabla^{\h}_{\x'} \x' = \left( \frac{a^2}{2} \, \grad^{\h} \! f \right) \delta_\eps(u). 
\label{xeqn-eps}
\end{gather}
\label{geodeqns-delta-eps}\end{subequations}
As observed above solvability of the system~\eqref{geodeqns-delta-eps} rests solely on the solvability of~\eqref{xeqn-eps}, since~\eqref{veqn-eps} is decoupled and linear and can therefore be solved by integration once $\x$ is known. 

\subsection{Completeness of rINPWs}

We now prove completeness of rINPWs for all `profile functions' $f$ in the following sense: Consider any geodesic $\gamma$ in the family of spacetimes $(M,\g_\eps)$ with initial data $\gamma(t_0)$, $\dot\gamma(t_0)$ at some parameter value $t_0$ (independent of the regularisation parameter $\eps$) such that $\gamma(t_0)$ lies outside of the support of $\delta_\eps$. Then we show that there exists an $\eps_0$ such that, for all $\eps \le \eps_0$, $\gamma$ is defined on all of $\R$ in any of the spacetimes $(M,\g_\eps)$. Note, however, that $\eps_0$ will, in general, depend on the data $\gamma(t_0)$, $\dot\gamma(t_0)$. The precise statement, in which we use the coordinate version of the geodesic equation~\eqref{geodeqns-delta-eps} and $\Gamma^{(N) k}_{ij}$ denotes the Christoffel symbols of $N$, is as follows.

\begin{Theorem}
\label{thm:reg}
Let $(N, \h)$ be a complete Riemannian manifold. Given data $(x_0,v_0) \in N \times \R$, $(\dot{x}_0,\dot{v}_0) \in T_{x_0}N \times \R$, there exists $\eps_0>0$ such that the initial value problem 
\label{eq:GEQ_proof}
\begin{align}\label{eq:ro-veq}
\ddot{v}_{\eps} &= -\delta_{\eps} \frac{\partial f}{\partial x^j}(x_{\eps}) \dot{x}^j_{\eps} -\frac{1}{2} f(x_{\eps}) \dot{\delta}_{\eps}\\
\ddot{x}^k_{\eps} &= - \Gamma^{(N) k}_{ij}(x_{\eps}) \, \dot{x}_{\eps}^i \, \dot{x}_{\eps}^j + \frac{1}{2} \delta_{\eps} h^{kl}(x_{\eps}) \frac{\partial f}{\partial x^l}(x_{\eps})
\label{eq:ro-xeq}
\end{align}
\begin{equation}\label{eq:ro-ide}
v_{\eps}(-1) = v_0, \quad \dot{v}_{\eps}(-1) = \dot{v}_0, \quad x_{\eps}(-1) = x_0, \quad \dot{x}_{\eps}(-1) = \dot{x}_0.
\end{equation}
has a unique solution on all of $\R$ for all $\eps \leq \eps_0$.
\end{Theorem}

The proof of this Theorem relies on a fixed point argument, which we formulate as Lemma~\ref{lem:IVP} below. The Lemma, in turn, relies on the following refinement of the Banach fixed-point theorem.

\begin{Lemma}[Weissinger's Fixed-Point Theorem~\cite{Wei:52}]
\label{thm:weissinger}
Let $(M, d)$ be a complete metric space, $X \subseteq M$ a nonempty, closed subset, and $(a_n)_n$ a sequence of positive real numbers with the property that $\sum_{n=1}^{\infty} a_n$ converges. Then a map $A \colon X \to X$ with the property that
\begin{equation}
\label{eq:Weissinger}
d(A^n(x),A^n(y)) \leq a_n \, d(x,y) \quad \forall x,y \in X, \quad \forall n \in \N,
\end{equation}
possesses a unique fixed point $\bar{x}$ in $X$, i.e.\ $\bar{x} = A(\bar{x})$.
\end{Lemma}

\begin{Lemma}
\label{lem:IVP}
Given functions $F_1 \in \mathcal{C}^{\infty}(\R^{2n},\R^n)$ and $F_2 \in \mathcal{C}^{\infty}(\R^n,\R^n)$, constants $b > 0$ and $c >0$, initial conditions $x_0, \dot{x}_0 \in \R^n$ and a strict delta net $(\delta_{\eps})_{\eps \in I}$ with $L^1$-bound $K>0$, the initial value problem
\begin{equation}
\label{eq:IVP_lemma}
\begin{split}
&\ddot{x}_{\eps} = F_1(x_{\eps},\dot{x}_{\eps}) + F_2(x_{\eps}) \delta_{\eps},\\
&x_{\eps}(-\eps) = x_0, \quad \dot{x}_{\eps}(-\eps) = \dot{x}_0, 
\end{split}
\end{equation}
has a unique solution $x_{\eps}$ on $J_{\eps} \coloneqq [-\eps,\alpha - \eps]$ with $(x_{\eps}(J_{\eps}),\dot{x}_{\eps}(J_{\eps})) \subseteq I_1 \times I_2$ where $I_1 \coloneqq \{x \in \R^n: |x -x_0| \leq b\}$, $I_2 \coloneqq \{x \in \R^n: |x -\dot{x}_0| \leq c + K \| F_2\|_{I_1,\infty}\}$ and
\begin{equation}
\label{eq:alpha}
\alpha \coloneqq \min \left(1, \frac{b}{|\dot{x}_0| + \|F_1\|_{I_1 \times I_2,\infty} + K \| F_2\|_{I_1,\infty}}, \frac{c}{\|F_1\|_{I_1 \times I_2,\infty} }\right).
\end{equation}
In particular, $x_{\eps}$ and $\dot x_\eps$ are uniformly bounded in $\eps$.
\end{Lemma}

The proof of this statement involves estimating the solution operator of the ODE~\eqref{eq:IVP_lemma}, 
\begin{equation}\label{intop}
A_{\eps}(x_{\eps})(t) \coloneqq x_0 + \dot{x}_0(t+{\eps}) + \int\limits_{-{\eps}}^t\int\limits_{-{\eps}}^s\!
F_1(x_{\eps}(r),\dot{x}_{\eps}(r))\,d r
\, d s + \int\limits_{-{\eps}}^t\int\limits_{-{\eps}}^s\! F_2(x_{\eps}(r))\delta_{\eps}(r)\, d r \,d s.
\end{equation}
One finds that this map is not a contraction mapping on the natural solution space, i.e. the subspace $X_{\eps} \coloneqq \{x_{\eps}\in\mathcal{C}^{\infty}(J_{\eps},\R^n):
x_{\eps}(J_{\eps})\subseteq I_1, \dot{x}_{\eps}(J_{\eps})\subseteq I_2\}$ of the Banach space ${\mathcal C}^1(J_{\eps},\R^n)$ with norm $\|x \|_{\mathcal{C}^1}=\|x\|_{J_{\eps},\infty}+\|\dot x\|_{J_{\eps},\infty}$. 
However, it can be shown that it satisfies the less restrictive assumptions of Weissinger's theorem. For details, we refer to~\cite[Lem.\ 5.3]{Frauenberger},~\cite[Lem.\ A2]{SS:12}. 

\medskip
The strength of Lemma~\ref{lem:IVP} is that the time of existence of the solution, $\alpha$, is independent of the regularisation parameter $\eps$. It is precisely this fact that allows us to extend the solutions beyond the `regularisation strip' given by the support of $\delta_\eps$ provided that the latter is narrower than $\alpha$. 

\medskip
With this remark in mind, we proceed to the proof of the theorem.

\begin{proof}[Proof of Theorem~\ref{thm:reg}]
We need only deal with the $x$-equation~\eqref{eq:ro-xeq}. For simplicity, we will assume (without loss of generality) that $\supp (\delta_{\eps}) \subseteq [-\eps,\eps]$, then $\delta_{\eps}(u) = 0$ for $u \in [-1,-\eps]$, and the $x$-equation simplifies to the geodesic equation in $(N, \h)$
\begin{equation}
\ddot{x}^k_{\eps} = - \Gamma^{(N) k}_{ij}(x_{\eps}) \, \dot{x}_{\eps}^i \, \dot{x}_{\eps}^j, 
\quad x_{\eps}(-1) = x_0, \quad \dot{x}_{\eps}(-1) = \dot{x}_0,
\end{equation}
By completeness of $(N, \h)$, this problem has a unique solution $x^N_{\eps} \colon \R \to N$ which is also the solution of the $x$-equation~\eqref{eq:ro-xeq} for $u \in (-\infty, -\eps]$, i.e.\ before it reaches the regularisation zone of the wave, see Figure~\ref{fig:wz}. What remains to be found is a solution to the initial value problem
\begin{equation}
\label{eq:GEQ_N_x_proof}
\begin{gathered}
\ddot{x}^k_{\eps} = - \Gamma^{(N) k}_{ij}(x_{\eps}) \, \dot{x}_{\eps}^i \, \dot{x}_{\eps}^j + \frac{1}{2} \delta_{\eps} h^{kl}(x_{\eps}) \frac{\partial f}{\partial x^l}(x_{\eps}),\\
\quad x_{\eps}(-\eps) = x^N_{\eps}(-\eps), \quad \dot{x}_{\eps}(-\eps) = \dot{x}^N_{\eps}(-\eps).
\end{gathered}
\end{equation}
If we can establish the existence of such a solution $\hat x_\eps$ until $u = \eps$ then, for $u \in [\eps, \infty)$ (i.e.\ after the regularisation zone) the existence of a unique solution $\tilde{x}^N_{\eps} \colon \R \to N$ with the appropriate initial conditions $\tilde{x}^N_{\eps}(\eps)=x_\eps(\eps)$, $\dot{\tilde{x}}^N_{\eps}(\eps)=\dot{x}_{\eps}(\eps)$ is again guaranteed by completeness of $(N, \h)$. 

The existence of such a solution $\hat x_\eps$, however, is achieved by applying Lemma~\ref{lem:IVP} to~\eqref{eq:GEQ_N_x_proof} with $F_1(x_{\eps},\dot{x}_{\eps}) = - \Gamma^{(N) k}_{ij}(x_{\eps}) \, \dot{x}_{\eps}^i \, \dot{x}_{\eps}^j$, $F_2(x_{\eps}) = \frac{1}{2} h^{kl}(x_{\eps}) \frac{\partial f}{\partial x^l}(x_{\eps})$, and arbitrary $b,c > 0$. Therefore, we obtain a unique solution $\hat{x}_{\eps} \colon [-\eps, \alpha - \eps] \to N$. Since $\alpha$ is independent of $\eps$, we may set $\eps_0 = \frac{\alpha}{2}$ and obtain the desired solution on all of $\R$ by defining 
\begin{equation*}
x_{\eps}(u) = 
\begin{cases}
x^N_{\eps}(u),& \text{for } u \in (-\infty,-\eps]\\
\hat{x}_{\eps}(u),& \text{for } u \in (-\eps,\eps)\\
\tilde{x}^N_{\eps}(u),& \text{for } u \in [\eps,\infty)
\end{cases}
\end{equation*}
for all $\eps \leq \eps_0$.

\begin{figure}\centering
\begin{tikzpicture}[baseline={([yshift=-.8ex]current bounding box.center)},scale=0.8]
\fill [black!80,opacity=0.2] (2,1) --(4,1) -- (4,6) -- (2,6);
\draw[dotted] (2,1) node[below]{$-\eps$} -- (2,6);
\draw[dotted] (4,1) node[below]{$\eps$} -- (4,6);
\draw[dotted] (3,1) node[below]{$0$} -- (3,6);
\draw (0,6) -- (1,5) node[above,xshift = 2]{$x^N_{\eps}$} -- (2,4); 
\draw[dashed] (2,4) to[out=-45,in = 90] (3,3) node[right]{$\hat{x}_{\eps}$} to[out=-90,in=153] (4,2);
\draw (4,2) -- (5,1.5) node[above]{$\tilde{x}^N_{\eps}$} (6,1);
\fill[black] (2,4) circle (2pt);
\fill[black] (4,2) circle (2pt);
\end{tikzpicture}
\caption{The solution $x_\eps$ of the ODE~\eqref{eq:ro-xeq} is given by background solutions on $(N,\h)$ outside the regularisation zone given by the support of the strict delta net $\delta_\eps$ which is displayed in gray.}
\label{fig:wz}
\end{figure}
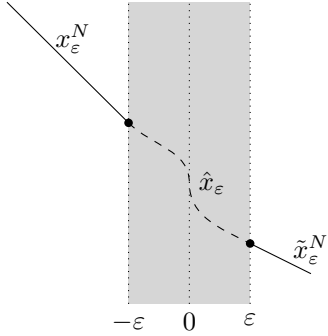

Finally we find the $v$-component of the solution for all $u\in\R$ by integration,
\[
v_{\eps}(u) = v_0 + \dot{v}_0 (u+1) + \int_{-1}^{u} \int_{-1}^t \biggl(-\delta_{\eps}(s) \frac{\partial f}{\partial x^j}(x_{\eps}(s)) \dot{x}^j_{\eps}(s) -\frac{1}{2} f(x_{\eps}(s)) \dot{\delta}_{\eps}(s)\biggr) \, ds \, dt. 
\]
This completes the proof.
\end{proof}

As mentioned at the beginning of this subsection, we have shown completeness of any chosen geodesic in the spacetimes $(M, \g_\eps)$ provided $\eps$ is sufficiently small, where the smallness required is determined by the data of the geodesic. Note that we have not shown that any of the spacetimes $(M, \g_\eps)$ for a fixed $\eps$ is complete. However, we shall now show that one has a more natural completeness result using nonlinear distributional geometry. 

\subsection{Nonlinear distributional geometry}

We now give a short but self-contained account of nonlinear distributional Lorentzian geometry. (See~\cite{KS:02a,KS:02b} and~\cite[Section 3.2]{GKOS:01}) for details.) The construction is based on Colombeau's (special) algebra of generalised functions~\cite{Col:85}, which provides an extension of Schwartz's theory of distributions~\cite{FJ:98} that has maximal consistency with respect to classical analysis. It is based on regularisation of distributions via nets of smooth functions and asymptotic estimates in terms of a regularisation parameter. 

We first define $\E(M)$ to be the set of all nets of smooth functions $(u_\eps)_{\eps\in \I}$ on a manifold $M$ that depend smoothly on $\eps$. The~\emph{algebra of generalised functions on $M$\/} is defined as the quotient $\G(M) \coloneqq 
\EM(M)/\NN(M)$ of~\emph{moderate\/} modulo~\emph{negligible\/} nets in $\E(M)$, defined via 
\[
\EM(M) \coloneqq \{ (u_\eps)_\eps\in\E(M):\, \forall K\Subset M\
\forall P\in{\mathcal P}\ \exists N:\ \sup\limits_{p\in 
 K}|Pu_\eps(p)|=O(\eps^{-N}) \}\,,\\
\]
\[\NN(M) \coloneqq \{ (u_\eps)_\eps\in\EM(M):\ \forall K\Subset M\
\forall m:\ \sup\limits_{p\in K}|u_\eps(p)|=O(\eps^{m}) \}\,.
\]
Here $K \Subset M$ means that $K$ is compact and ${\mathcal P}$ denotes the space of linear differential operators on $M$. Elements $u = [(u_\eps)_\eps]$ of $\G(M)$ are referred to as~\emph{generalised functions with representative $(u_\eps)_\eps$\/} and sums, products, and the Lie derivative of $u$ are defined componentwise (i.e., for fixed $\eps$). It can be shown that $\G(M)$ is a fine sheaf of differential algebras. \emph{Generalised tensor fields\/} on $M$ are defined as tensor fields with $\G(M)$-coefficients, i.e.
\begin{equation}
\G^r_s(M)=\G(M)\otimes_{\Cinf(M)}{\mathcal T}^r_s(M).
\end{equation}
Here ${\mathcal T}^r_s(M)$ is the space of smooth $(r,s)$-tensor fields on $M$ and $\G^r_s(M)$ is a fine sheaf of finitely generated and projective $\G$-modules.

Distributions $u\in\D'(M)$ can be embedded into $\G(M)$ via convolution (as in~\eqref{eq:conv}) giving rise to a sheaf homomorphism, $\iota$, that is consistent with the product of smooth functions. A coarser way of relating generalised functions to distribution is as follows: We say that $u, v\in \G(M)$ are~\emph{associated\/} to one another, $u \approx v$, if $u_\eps - v_\eps \to 0$ weakly. We say that $w\in \D'(M)$ is~\emph{associated with $u$\/} if $u\approx \iota(w)$.
\medskip

One benefit of this framework is that one may work componentwise (i.e. for fixed $\eps$) as in the classical smooth setting. In particular, since products are well defined, one may insert generalised vector fields and one-forms into generalised tensor fields, which is not possible in the distributional setting, cf.~\cite{Mar:67}. More explicitly we adopt the following definitions.

\begin{definition}[Generalised spacetime]
We call a $(0,2)$-tensor field $\g\in\G^0_2(M)$ a~\emph{generalised Lorentzian metric\/} if, for all small $\eps$, one (hence any) representative $\g_\eps$ consists of Lorentzian metrics and there is $m$ such that $|\det (\g_\eps)_{ij}| \geq \eps^m$ on compact sets.%
\footnote{This nondegeneracy condition is equivalently to $\det g$ being invertible in $\G$.} %
If, in addition, for all small $\eps$ one (hence any) $\g_\eps$ can be time oriented by the same smooth vector field, we call $(M, \g)$ a~\emph{generalised spacetime}.
\end{definition}

Any generalised metric induces an isomorphism between generalised vector fields and one-forms, and there is a unique Levi-Civita connection $\nabla$ corresponding to $\g$. Moreover, the curvature tensor and all its contractions can be calculated on the level of representatives using the usual expressions. Finally, the curvature defined in this way consistently extends the distributional curvature of the Geroch--Traschen setting~\cite{GT:87,LFM:07}, see~\cite{SV:09,S:08}. 
\medskip

To define the notion of geodesics of a generalised metric we need to introduce the space of generalised curves $\G[J,M]$ on an interval $J$~\emph{taking values in $M$}. Again we use a quotient of moderate modulo negligible nets $(\gamma_{\eps})_{\eps}$ of smooth curves, where moderateness (negligibility) is defined by moderateness (negligibility) of $(\psi\circ \gamma_\eps)_\eps$ for all smooth $\psi \colon M\to\R$. Moreover, $(\gamma_\eps)_\eps$ is required to be~\emph{c-bounded}, i.e., $\gamma_\eps([a,b])$ is contained in some $L \Subset M$ for $\eps$ small and all compact subintervals $[a, b]$ of $J$. Observe that no distributional analogue of such a space of curves exists and it has long been realised that regularisation is a possible remedy, cf.~\cite{Mar:67}.

The~\emph{induced covariant derivative\/} of a generalised vector field $\xi=[(\xi_\eps)_\eps]\in\G^1_0(M)$ on a generalised curve $\gamma=[(\gamma_\eps)_\eps]\in\G[J,M]$ is defined componentwise and yields a generalised vector field $\xi'$ on $\gamma$. A~\emph{geodesic\/} in a generalised spacetime is then a curve $\gamma\in\G[J,M]$ satisfying $\gamma''=0$. Equivalently we have the usual local expressions, i.e., 
\begin{equation}\label{geo}
\Big[\,\Big(\frac{d^2\gamma_\eps^k}{d\lambda^2}
+\sum_{i,j}{\Gamma_\eps}^k_{ij}\frac{d\gamma_\eps^i}{d\lambda}\frac{d\gamma_\eps^j}{d\lambda}\Big)_\eps\,\Big]
=0,
\end{equation}
where $\Gamma^k_{ij}=[({\Gamma_\eps}^k_{ij})_\eps]$ denote the Christoffel symbols of $\g=[(\g_\eps)_\eps]$. Finally we define following~\cite[Def.~2.1]{SS:15}:
\begin{definition}[Geodesic completeness]
A generalised spacetime $(M, \g)$ is~\emph{geodesically complete\/} if any of its geodesics $\gamma$ can be defined on all of $\R$, i.e. $\gamma\in\G[\R,M]$.
\end{definition}

\subsection{INPWs as generalised spacetimes \& their geodesic completeness}

We next define the INPWs as generalised spacetimes. We refer to a generalized function $D \in \mathcal{G}(\R)$ that can be represented by a strict delta net $(\delta_{\eps})_{\eps \in (0,1]}$, i.e., $ D =[(\delta_{\eps})_{\eps}] \in \mathcal{G}(\R)$ , as a~\emph{generalized delta function}.

\begin{definition}[INPW]
\label{def:INPW}
A generalized spacetime $(M,\g)$ is called an \emph{impulsive NPW} or \emph{INPW} if $M = N \times \R^2_1$, $(N, \h)$ an $(n-2)$-dimensional Riemannian manifold and its metric is of the form
\begin{equation}
\label{eq:gINPW_metric}
\g = 2 du \, dv + f(x)D(u) \, du^2 + \h,
\end{equation}
where $D\in \mathcal{G}(\R)$ is a generalized delta function.
\end{definition}

\begin{remark}
Recall that the embedding $\iota$ via convolution mentioned above depends on several choices such as the $\delta_\eps$, a selection of charts, as well as a partition of unity and a family of cut-off functions~\cite[Thm.\ 3.2.10]{GKOS:01}. Therefore, instead of embedding an INPW into $\G^0_2(M)$ via $\iota$, we have chosen to define it via the entire class of strict $\delta$-nets, which is diffeomorphism invariant.
\end{remark}

Obviously a generalised $\g$ which is an INPW has an rINPW $\g_\eps$ as a representative. Consequently the geodesic equations for such a $\g$ are given on the level of representatives by~\eqref{geodeqns-delta-eps}. Before exploiting this fact, we briefly explain what it means to solve a differential equation in generalised functions.

To establish existence and uniqueness of a solution $\gamma \in \mathcal{G}[J,M]$ to the geodesic equation $\gamma''=0$ with given initial data $\gamma(t_0) = x_0$, $\gamma'(t_0) = x'_0$ (with $x_0$, $x'_0\in\R$) one proceeds in three steps. First, one needs to provide a~\emph{solution candidate}, i.e., a net of smooth functions $\gamma_{\eps} \colon J \to M$ solving the corresponding initial value problem on the level of representatives. (In our case, written in local coordinates, the equations~\eqref{eq:ro-veq}, \eqref{eq:ro-xeq} with~\eqref{eq:ro-ide}.) %
Secondly, we must check that the solution candidate $\gamma_{\eps} \colon J \to M$ is a representative of a generalised curve $\gamma = [(\gamma_{\eps})_{\eps}] \in \mathcal{G}[J,M]$, i.e. one needs to establish that $(\gamma_{\eps})_{\eps}$ is compactly bounded and moderate. %
Thirdly, to show uniqueness, we need to establish that solutions to the negligibly perturbed initial value problem differ only negligibly from $(\gamma_{\eps})_{\eps}$. 

Since we are ultimately interested in a completeness result, we must also establish that the solutions to the geodesic equation are actually defined on all of $\R$. This will follow from the fact that there are solution candidates defined on $\R$ as asserted by Theorem~\ref{thm:reg}. In total, we will prove the following result:

\begin{Theorem}[Completeness of INPWs]
\label{thm:comp}
For any impulsive $N$-fronted wave with parallel rays $(M,\g)$ with complete Riemannian part $(N, \h)$ the initial value problem for the geodesics is uniquely solvable in $\mathcal{G}[\R,M]$. In particular, any INPW is geodesically complete. 
\end{Theorem}

\begin{proof}
Given initial data, by Theorem~\ref{thm:reg}, we can find $\eps_0$ such that the regularised initial value problem~\eqref{eq:ro-veq},~\eqref{eq:ro-xeq},~\eqref{eq:ro-ide} has a unique (smooth) solution $(x_{\eps},v_{\eps})$ that is defined on all of $\R$. We may extend this net arbitrarily yet smoothly for $\eps_0 < \eps \leq 1$ to obtain a net in $\mathcal{C}^{\infty}(\R,N\times\R)^{(0,1]}$ as a solution candidate.

According to the above discussion, we need to show that:
\begin{itemize}
\item[(a)] $[(x_{\eps},v_{\eps})_{\eps}]$ is an element of $\mathcal{G}[\R,N\times \R]$, i.e. $(x_{\eps},v_{\eps})_{\eps}$ is compactly bounded and moderate; 
\item[(b)] negligibly perturbing the equations~\eqref{eq:ro-veq},~\eqref{eq:ro-xeq},~\eqref{eq:ro-ide} leads to solutions that differ negligibly from the solution candidate $(x_{\eps},v_{\eps})_{\eps}$.
\end{itemize}

\noindent
(a) By Lemma~\ref{lem:IVP}, $x_\eps$ is uniformly bounded in $\eps$ on the regularisation zone and, by construction, independent of $\eps$ outside it. The same holds true for $\dot x_\eps$, so both are of order $O(1)$ in $\eps$. This allows us to use the differential equation~\eqref{eq:ro-xeq} to see that $\ddot x_\eps$ is $O(1/\eps)$ since the $\delta$-net is of that order. Now inductively differentiating~\eqref{eq:ro-xeq} we obtain order $O(1/\eps^n)$-estimates for the higher derivatives of $x_\eps$. Therefore, the net $(x_\eps)_\eps$ is c-bounded and moderate. Then $v_{\eps}$ is given by integrating~\eqref{eq:ro-veq}, i.e. by the integral of moderate and compactly bounded nets, and hence is moderate and compactly bounded itself. Therefore, $[(x_{\eps},v_{\eps})_{\eps}]$ is an element of $\mathcal{G}[\R,N\times\R]$.
\medskip

\noindent
(b) To prove uniqueness of solutions in $\mathcal{G}[\R,N\times\R]$ we assume
that there is another solution $(\tilde{x},\tilde{v}) \in \mathcal{G}[\R,N \times \R]$ of the initial avlue problem, i.e., that there is a net $(\tilde{x}_{\eps},\tilde{v}_{\eps})_{\eps}$ representing $(\tilde{x},\tilde{v}) = [(\tilde{x}_{\eps},\tilde{v}_{\eps})_{\eps}]$ and solving the initial value problem
\begin{equation}
\label{eq:GEQ_reg_tilde}
\begin{gathered}
\ddot{\tilde{v}}_{\eps} = -\delta_{\eps} \frac{\partial f}{\partial x^j}(\tilde{x}_{\eps}) \dot{\tilde{x}}^j_{\eps} -\frac{1}{2} f(\tilde{x}_{\eps}) \dot{\delta}_{\eps} + a_{\eps}\\
\ddot{\tilde{x}}^k_{\eps} = - \Gamma^{(N) k}_{ij}(\tilde{x}_{\eps}) \, \dot{\tilde{x}}_{\eps}^i \, \dot{\tilde{x}}_{\eps}^j + \frac{1}{2} \delta_{\eps} h^{kl}(\tilde{x}_{\eps}) \frac{\partial f}{\partial x^l}(\tilde{x}_{\eps}) + b_{\eps}^k\\
\tilde{v}_{\eps}(-1) = v_0 + c_{\eps}, \quad \dot{\tilde{v}}_{\eps}(-1) = \dot{v}_0 + \dot{c}_{\eps}, \quad \tilde{x}_{\eps}(-1) = x_0 + d_{\eps}, \quad \dot{\tilde{x}}_{\eps}(-1) = \dot{x}_0 + \dot{d}_{\eps},
\end{gathered}
\end{equation}
where $(a_{\eps})_{\eps}$, $(b_{\eps})_{\eps} \in \mathcal{N}[\R,N]$ are negligible functions and $(c_{\eps})_{\eps}$, $(\dot{c}_{\eps})_{\eps}$, $(d_{\eps})_{\eps}$, $(\dot{d}_{\eps})_{\eps}$ are negligible constants, i.e., $|c_\eps|=O(\eps^m)$ for any $m$ and likewise for the other constants.

We need to show that both $(x_{\eps}-\tilde{x}_{\eps})_{\eps}$ and $(v_{\eps}-\tilde{v}_{\eps})_{\eps}$ are negligible.
To this end, we first estimate $(x_{\eps}-\tilde{x}_{\eps})$ and $(\dot{x}_{\eps}-\dot{\tilde{x}}_{\eps})$ using the integral versions of the differential equations in the notation of Lemma~\ref{lem:IVP} for $F_1(x,\dot{x}) = - \Gamma^{(N) k}_{ij}(x) \, \dot{x}^i \, \dot{x}^j$ and $F_2(x) = \frac{1}{2} h^{kl}(x) \frac{\partial f}{\partial x^l}(x)$. We first observe that the terms arising from the negligible perturbation $b_\eps$ only give rise to negligible terms. More precisely, we have that $\forall T>0$, $\forall q \in \N$, there are constants $K_1, K_2 >0$ such that, for sufficiently small $\eps$ and all $t \in [-T,T]$, 
\begin{equation}
|d_{\eps}| + |t \, \dot{d}_{\eps}| + \left|\int_{-1}^t \int_{-1}^s b^k_{\eps}(r) \, dr \, ds \right| \leq K_1 \, \eps^q,
\quad
|\dot{d}_{\eps}| + \left|\int_{-1}^t b^k_{\eps}(s)\, ds \right| \leq K_2 \, \eps^q.
\end{equation}
So we find again for all $T>0$ and all exponents $q \in \N$ that there are constants $K_1$, $K_2 >0$ such that for all small $\eps$ and all $t \in [-T,T]$
\begin{equation*}
\begin{split}
|x_{\eps}(t)-\tilde{x}_{\eps}(t)| \leq& K_1 \eps^q + \int_{-1}^t \int_{-1}^s |F_1(x_{\eps}(r),\dot{x}_{\eps}(r)) - F_1(\tilde{x}_{\eps}(r),\dot{\tilde{x}}_{\eps}(r))|\, dr \, ds \\
& \qquad \qquad + \int_{-1}^t \int_{-1}^s |F_2(x_{\eps}(r)) - F_2(\tilde{x}_{\eps}(r))| \, |\delta_{\eps}(r)|\, dr \, ds \\ \leq
& K_1 \eps^q + L_1 \int_{-1}^t \int_{-1}^s |x_{\eps}(r) - \tilde{x}_{\eps}(r)| + |\dot{x}_{\eps}(r) - \dot{\tilde{x}}_{\eps}(r)|\, dr \, ds \\
& \qquad \qquad + L_2 \int_{-1}^t \int_{-1}^s |x_{\eps}(r) - \tilde{x}_{\eps}(r)| \, |\delta_{\eps}(r)|\, dr \, ds ,
\end{split}
\end{equation*}
and
\begin{equation*}
\begin{split}
|\dot{x}_{\eps}(t)-\dot{\tilde{x}}_{\eps}(t)| \leq& K_2 \eps^q + \int_{-1}^t |F_1(x_{\eps}(s),\dot{x}_{\eps}(s)) - F_1(\tilde{x}_{\eps}(s),\dot{\tilde{x}}_{\eps}(s))|\, ds \\
& \qquad \qquad + \int_{-1}^t |F_2(x_{\eps}(s)) - F_2(\tilde{x}_{\eps}(s))| \, |\delta_{\eps}(s)| \, ds \\ \leq
& K_2 \eps^q + L_1 \int_{-1}^t |x_{\eps}(s) - \tilde{x}_{\eps}(s)| + |\dot{x}_{\eps}(s) - \dot{\tilde{x}}_{\eps}(s)| \, ds \\
& \qquad \qquad + L_2 \int_{-1}^t |x_{\eps}(s) - \tilde{x}_{\eps}(s)| \, |\delta_{\eps}(s)| \, ds .
\end{split}
\end{equation*}
Adding these two inequalities and defining $\psi(t) \coloneqq |x_{\eps}(s) - \tilde{x}_{\eps}(s)| + |\dot{x}_{\eps}(s) - \dot{\tilde{x}}_{\eps}(s))|$ we have
\begin{equation*}
\psi(t) \leq (K_1 +K_2) \eps^q + \int_{-1}^t (L_1 +L_2|\delta_{\eps}(s)|)\psi(s) \, ds + \int_{-1}^t \int_{-1}^s (L_1+L_2 |\delta_{\eps}(r)|)\psi(r)\, dr \, ds. 
\end{equation*}
It then follows from Bykov's extension of Gronwall's inequality~\cite[Thm.~11.1]{BS:92} that
\[
\psi(t) \leq (K_1 +K_2) \eps^q \exp\biggl( \int_{-1}^t (L_1 +L_2|\delta_{\eps}(s)|) \, ds + \int_{-1}^t \int_{-1}^s (L_1+L_2 |\delta_{\eps}(r)|)\, dr \, ds \biggr) \leq K_3 \eps^q,
\]
with $K_3$ explicitly given by $K_3= (K_1 + K_2)\exp((T+1)L_1 + KL_2 + (T+1)^2 L_1 + (T+1)KL_2)$ where $K$ is the $L^1$-bound of $(\delta_{\eps})_{\eps}$.
This establishes negligibility of $(x_{\eps}-\tilde{x}_{\eps})_{\eps}$ by~\cite[Thm. 1.2.3]{GKOS:01}. Since $(v_{\eps}-\tilde{v}_{\eps})$ is essentially obtained by integrating $(x_{\eps}-\tilde{x}_{\eps})$, we conclude that it is also negligible.
\end{proof}

Finally, we remark that the above result establishes the failure of the Ehlers--Kundt conjecture in the impulsive case.

\begin{Corollary}
The Ehlers--Kundt conjecture fails for INPWs which are geodesically complete irrespectible of the (smooth) profile function $f$.
\end{Corollary}

Conceptually, completeness ceases to be a rigidity criterion in the presence of a distributional wave profile. This not only settles the natural ``impulsive EK'' analogue in the negative, but also underscores the subtle role played by regularity in the geometric and dynamical characterizations of exact wave spacetimes.%

\subsection*{Acknowledgment}
This article is based on the Masters thesis of the first author~\cite{Frauenberger}. J.G and R.S. are very grateful to Annegret Burtscher, Jos\'e Luis Flores, Lilia Mehidi, and A. Shadi Tahvildar-Zadeh for providing such an engaging atmosphere during the BIRS-IMAG workshop ``Geometry, Analysis, and Physics in Lorentzian Signature'' held in Granada, Spain in May 2025.

This research was funded in part by the Austrian Science Fund (FWF) [Grant DOI 10.55776/EFP6]. For open access purposes, the authors have applied a CC BY public copyright license to any author accepted manuscript version arising from this submission.

%
%
%

\providecommand{\bysame}{\leavevmode\hbox to3em{\hrulefill}\thinspace}



\end{document}